\documentclass[12pt, a4paper]{article}
\usepackage[utf8]{inputenc}
\usepackage{authblk}

\title{$\mathcal{O}(k)$-robust spanners in one dimension}
\author[1]{Kevin Buchin}
\author[1]{Tim Hulshof}
\author[1]{D\'aniel Ol\'ah}
\affil[1]{Eindhoven University of Technology, Eindhoven, The Netherlands}
\date{}

\usepackage{graphicx}
\usepackage{geometry}

\usepackage{amsmath}
\usepackage{amsfonts}
\usepackage{amsthm}
\usepackage{amssymb}

\theoremstyle{plain}
\newtheorem{theorem}{Theorem}

\begin{document}

\maketitle

\begin{abstract}

A geometric $t$-spanner on a set of points in Euclidean space is a graph containing for every pair of points a path of length at most $t$ times the Euclidean distance between the points.
Informally, a spanner is $\mathcal{O}(k)$-robust if deleting $k$ vertices only harms $\mathcal{O}(k)$ other vertices. 
We show that on any one-dimensional set of $n$ points, for any $\varepsilon>0$, there exists an $\mathcal{O}(k)$-robust $1$-spanner with $\mathcal{O}(n^{1+\varepsilon})$ edges. 
Previously it was only known that $\mathcal{O}(k)$-robust spanners with $\mathcal{O}(n^2)$ edges exists and that there are point sets on which any $\mathcal{O}(k)$-robust spanner has $\Omega(n\log{n})$ edges. 

\end{abstract}

\section{Introduction}

\textit{Geometric networks} are graphs whose vertices are points in the $d$-dimensional Euclidean space $\mathbb{R}^d \;(d\geq 1)$ and whose edges are weighted by the Euclidean distance between their endpoints. A desirable property of geometric networks is to have relatively short paths between any pair of points. We say that a geometric network $G=(V,E)$ is a $t$\textit{-spanner} of $V'\subseteq V$ for some $t\geq 1$ if 
\begin{equation}
    d_G(x,y) \leq t \cdot d(x,y) \quad \quad \forall x,y \in V'\;,
\end{equation}
where $d_G(x,y)$ denotes the length of the shortest path between $x$ and $y$ in $G$ and $d(x,y)$ is the Euclidean distance between $x$ and $y$. If $V'$ is not specified, then we mean $V'=V$. Spanners have been studied extensively and there are many algorithms to construct spanners with various properties. Narasimhan and Smid~\cite{spanner_book} give a comprehensive overview of spanners.

An interesting property of spanners is their resistance against failures. Assume that some points in the network fail, that is, these points with all their edges are removed from the graph. The first natural question is whether the remaining part of the network is still a $t$-spanner or not, perhaps for a larger value of $t$. It is obvious that any point can be isolated by deleting its neighbors, so the remaining graph cannot in all cases remain a spanner. 

Bose et al.~\cite{RGS} introduced the following notion of robustness. 
Let $G=(V,E)$ be a (geometric) $t$-spanner for some $t\geq1$ and let $f:\mathbb{N}\rightarrow \mathbb{N}$ be an arbitrary function. Assume that a set of points $F\subset V$ fail. Then we say that $G$ is an $f(k)$-\textit{robust} $t$-\textit{spanner} if for any $F$, there exists a set $F^*\supseteq F$ with $|F^*|\leq f(|F|)$,
such that the subgraph induced by $V\backslash F$ is a $t$-spanner of $V\backslash F^*$. This means that it is allowed to ignore further points in the sense that we do not require the spanner property for them, but they are not deleted from $G$.
An alternative and stronger definition could be that the subgraph induced by $V\backslash F^*$ is a $t$-spanner. That is, we delete not only the failed vertices from $G$ but all the vertices that are in $F^*$.

Bose et al.~\cite{RGS} proved various bounds on the size (number of edges) of $f(k)$-robust spanners in one and higher dimensions for general functions $f(k)$. In particular, they prove that an $\mathcal{O}(k)$-robust spanner on $n$ vertices may needs $\Omega(n\log{n})$ edges, even for one-dimensional point sets. They show no upper bounds for the size of $\mathcal{O}(k)$-robust spanners below the trivial $\mathcal{O}(n^2)$. As open problem they pose obtaining tighter bounds for the size of $\mathcal{O}(k)$-robust spanners, even in the simple setting where the input is the set $\{1, 2, \ldots, n\} \subset \mathbb{R}^1$. In this paper, we address this problem for general one-dimensional point sets.

The definition of $f(k)$-robustness is not the only way to obtain failure resistant networks. An alternative concept is \textit{$k$-fault tolerance}~\cite{cz-ftgs-04,lns-iafts-02}. A graph $G$ is a $k$-fault tolerant $t$-spanner if for any set of points $F$, where $|F|\leq k$, the graph $G$ after the removal of the points of $F$ is a $t$-spanner. Fault tolerance is a suitable definition if the number of failures is bounded by some constant. In many cases, however, the number of failures is not known in advance and might be large. 
Clearly, the degree of each vertex must be at least $k+1$ to achieve $k$-fault tolerance, therefore the size of the graph is large if $k$ is big.
The definition of $f(k)$-robustness allows us to avoid this, at the cost of possibly having to ignore a small number of additional vertices.

In Section~\ref{non-it} we present a simple construction for which we show that it is an $\mathcal{O}(k)$-robust $1$-spanner on $n$ vertices with $\mathcal{O}(n^\frac{3}{2})$ edges for any one-dimensional point set. In Section~\ref{it-section} we further improve the upper bound by generalizing the construction given in Section~\ref{non-it}. We prove that for any one-dimensional point set there exists an $\mathcal{O}(k)$-robust $1$-spanner of size $\mathcal{O}(n^{1+\varepsilon})$.

\section{A simple construction}
\label{non-it}

Let $V=\{x_1,x_2,\dots,x_n\}$ be an arbitrary point set with $x_i\in \mathbb{R}$ for $1\leq i\leq n$ and $x_i<x_{i+1}$ for $1\leq i\leq n-1$. For the sake of simplicity, assume, that $n=(2m)^2$, where $m$ is a positive integer. We construct a graph $G=(V,E)$ as follows. Let $\mathcal{C}_i=\{ x_{(i-1)m +1}, x_{(i-1)m+2}, \dots, x_{(i+1)m} \}$ for $1\leq i \leq 4m-1$. We call $\mathcal{C}_i$ the $i^{th}$ cluster. There are $4m-1$ clusters and each of them has exactly $2m$ points. Note that adjacent clusters are half-overlapping, that is $|\mathcal{C}_i \cap \mathcal{C}_{i+1}|=m$. 
We define half-clusters as the sets obtained by splitting each cluster in the middle. Therefore, the number of half-clusters we have is exactly $4m$ and each of them contains exactly $m$ points. 
Let $\mathcal{H}_i^L$ and $\mathcal{H}_i^R$ denote the left and right half of $\mathcal{C}_i$, respectively. Also note that $\mathcal{H}_i^L=\mathcal{H}_{i-1}^R$.
The structure of the clusters is illustrated in Figure~\ref{clusters-nonit}.

\begin{figure}
\centering
\includegraphics[scale=1.2]{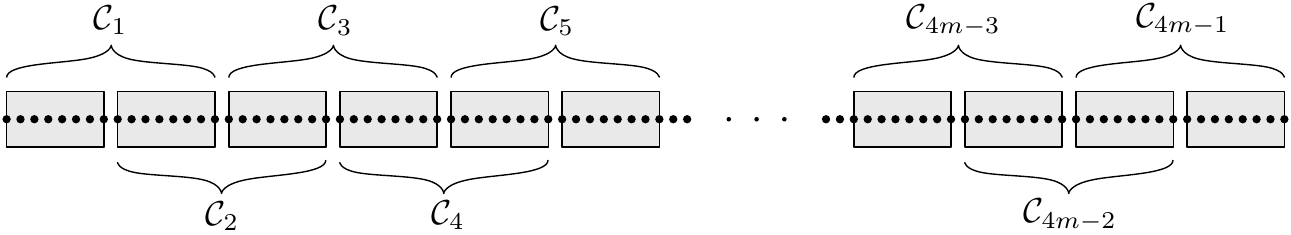}
\caption{The structure of the clusters of $G$. Grey rectangles indicate half-clusters. Each half-cluster contain exactly $m$ vertices.}
\label{clusters-nonit}
\end{figure}

We add two types of edges to the graph. First, we form cliques in each of the $4m-1$ clusters. Since the size of the clusters is $2m$, we add at most $(4m-1)\mathcal{O}(m^2)=\mathcal{O}(m^3)$ edges. Second, for any pair of half-clusters we add an arbitrary complete matching between the two half-clusters. The number of edges we add is again $\mathcal{O}(m^3)$, since each matching consists of exactly $m$ edges and there are $\genfrac{(}{)}{0pt}{1}{4m}{2}$ pairs of half-clusters.



\begin{theorem}
\label{nonit-proof}
For any $1$-dimensional point set $V$ (with $|V|=n$), the graph $G$ constructed above is an $\mathcal{O}(k)$-robust $1$-spanner with $\mathcal{O}(n^\frac{3}{2})$ edges.
\end{theorem}
\begin{proof}
Let's assume for now that $n=(2m)^2$; later we will show how to adapt the construction to arbitrary $n$. Clearly, the size of $G$ is $\mathcal{O}(m^3)=\mathcal{O}(n^\frac{3}{2})$. Firstly, it remains to construct the set $F^*$ for any set of failed points $F$ such that $|F^*|\leq\mathcal{O}(|F|)$. Secondly, we have to prove that $d_{G'}(x,y)=d(x,y)$ holds for any two points  $x,y\in V\backslash F^*$, where $G'$ is the subgraph induced by $V\backslash F$, i.e., $G'$ is a $1$-spanner for $V\backslash F^*$ on $V\backslash F$.

Consider the case that a set of points $F$ fail. 
To start, set $F^*=F$. Then for each half-cluster $\mathcal{H}_i^L$ ($2\leq i\leq 4m-1$), if $|\mathcal{H}_i^L \cap F| \geq \frac{m}{2}$, we update $F^*$ by adding the clusters $\mathcal{C}_{i-1}$ and $\mathcal{C}_i$ to $F^*$, 
see Figure~\ref{nonit-delrule}. 
Formally,
\begin{equation}
F^*:= F \cup \bigcup_{i: |\mathcal{H}_i^L \cap F|\geq\frac{m}{2}} \left(\mathcal{C}_{i-1} \cup \mathcal{C}_i \right).
\end{equation}
It is clear that $|F^*|\leq 6\cdot |F|$. 

\begin{figure}[b]
\centering
\includegraphics[scale=1.2]{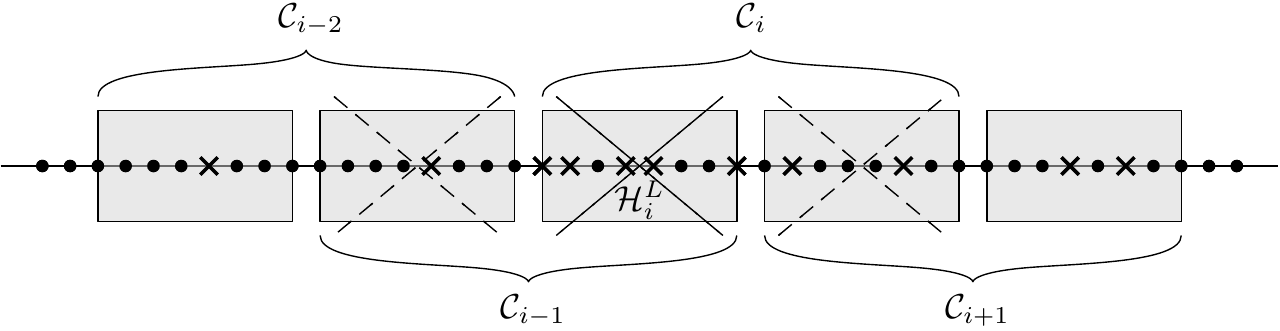}
\caption{Small crosses indicate failed nodes. Since $|\mathcal{H}_i^L\cap F|\geq \frac{m}{2}$, we add all points of cluster $\mathcal{C}_i$ and $\mathcal{C}_{i-1}$ to the set $F^*$. That is for all $x \in \mathcal{C}_i \cup \mathcal{C}_{i-1}$ let $x\in F^*$.}
\label{nonit-delrule}
\end{figure}

Fix $x,y\in V\backslash F^*$ and assume that $x<y$.
There are three separate cases. First, if they are in the same cluster, then $\{x,y\}\in E$ and the claim holds. 
Second, if they are not in the same, but overlapping clusters, then in the intersection of the clusters there is a vertex $z\notin F$, which shares an edge with both $x$ and $y$.
The third case is when $x\in \mathcal{H}_i^L$ and $y\in \mathcal{H}_{j}^R$, with $i<j+1$. Then, we know that $|\mathcal{H}_i^R \cap F| < \frac{m}{2}$, otherwise $x \in F^*$. Similarly, we know that $|\mathcal{H}_j^L \cap F| < \frac{m}{2}$, otherwise $y \in F^*$ holds.
Therefore, by the pigeonhole principle there is a vertex $x'\in\mathcal{H}_i^R \backslash F$ and a vertex $y'\in\mathcal{H}_j^L\backslash F$ for which $\{x',y'\}\in E$. It is clear that $\{x,x'\}\in E$ and $\{y,y'\}\in E$, since within a cluster all nodes are connected (Figure~\ref{nonit-1-path}). Therefore, the length of this path between $x$ and $y$ is exactly $d(x,x')+d(x',y')+d(y',y)=d(x,y)$, because $x<x'<y'<y$ by construction. 


\begin{figure}
\centering
\includegraphics[scale=1]{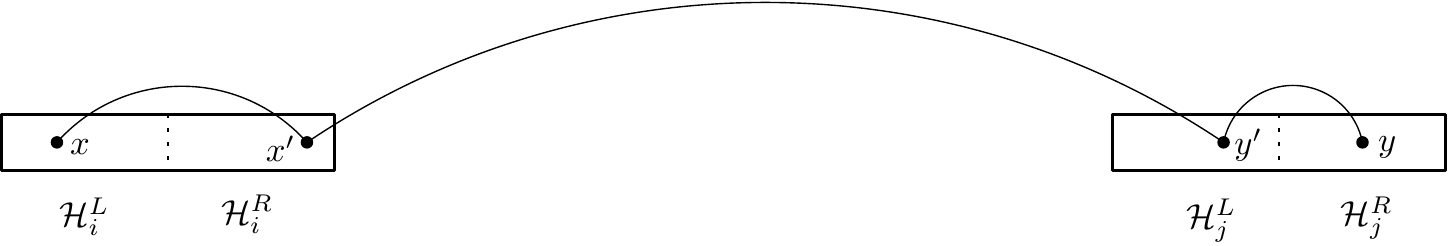}
\caption{The path of length $d(x,y)$ between the vertices $x$ and $y$.}
\label{nonit-1-path}
\end{figure}

Now we show the extension to arbitrary $n$. Let's assume that $(2m)^2<n< (2m+2)^2$. Split $V$ into two parts such that $V=V_1 \cup V_2$, where $V_1:=\{ x_1,x_2,\dots ,x_{(2m)^2} \}$ and $V_2:=\{ x_{(2m)^2+1},x_{(2m)^2+2},\dots ,x_n \}$. Build the same graph on the set $V_1$ 
and then extend the same structure to $V_2$. Since $|V_2|<8m+4$, at most $8$ full clusters can be added, which are half-overlapping and have size $2m$. If necessary, we add one more cluster that has a smaller right half that containing the points which were not added to any of the clusters yet. The size of the graphs is still $\mathcal{O}(n^\frac{3}{2})$ and construction of $F^*$ and the proof of the existence of paths without any detour remains the same.

\end{proof}

\section{Iterated construction}
\label{it-section}

Let $V$ be the same as before, and assume that 
$n=(2m)^{\ell+1}$, where $\ell$ is an additional parameter, which determines the number of layers in the construction of the spanner $G_\ell=(V,E)$. In each layer there are clusters with the same size, namely the clusters in layer $i$ have size $n^{\frac{i}{\ell+1}}=(2m)^i$ for $1\leq i \leq \ell$. Let $\mathcal{C}_{i,j}$ denote the $j^{th}$ cluster in layer $i$, then the clusters of the $i^{th}$ layer are
\begin{equation}
\mathcal{C}_{i,j}:=\Big\{ x_{(j-1)\frac{(2m)^i}{2}+1}, x_{(j-1)\frac{(2m)^i}{2}+2}, \dots, x_{(j+1)\frac{(2m)^i}{2}} \Big\} 
\end{equation}
for any $1\leq j\leq 2\cdot (2m)^{\ell+1-i}-1$ and $1\leq i\leq \ell$.
In each layer the clusters are half-overlapping, so that $|\mathcal{C}_{i,j}\cap \mathcal{C}_{i,j+1}|=\frac{(2m)^i}{2}$. Again, we define half-clusters by splitting each cluster into two parts of equal size. Let $\mathcal{H}_{i,j}^L$ and $\mathcal{H}_{i,j}^R$ denote the left and right half of $\mathcal{C}_{i,j}$, respectively. Note that a cluster in layer $i$ contains $2\cdot n^\frac{1}{\ell+1}=4m$ half-clusters from layer $i-1$. The case $\ell=1$ corresponds to the simple construction in Section~\ref{non-it}.


We define the edgeset of $G_\ell$ in the following way. For any cluster $\mathcal{C}_{1,j}$ in the first (lowest) layer, form a clique on its vertices. This adds $\mathcal{O}(n^\frac{2}{\ell+1}\cdot n^\frac{\ell+1-1}{\ell+1})= \mathcal{O}(n^\frac{\ell+2}{\ell+1})$ edges. In layer $i$ ($2\leq i\leq \ell$), for any cluster $\mathcal{C}_{i,j}$, form a complete matching between any pair of half-clusters from layer $i-1$ that are both contained in $\mathcal{C}_{i,j}$, see Figure~\ref{it-edges}. The number of edges that are added is $\mathcal{O}(n^\frac{\ell+1-i}{\ell+1}\cdot n^\frac{2}{\ell+1}\cdot n^\frac{i-1}{\ell+1})=\mathcal{O}(n^\frac{\ell+2}{\ell+1})$ in each layer. Finally we add a complete matching between any pair of half-clusters in the top level. This adds another $\mathcal{O}(n^\frac{\ell+2}{\ell+1})$ edges. Thus, $G_\ell$ has $\mathcal{O}(\ell\cdot n^\frac{\ell+2}{\ell+1})$ edges.

\begin{figure}
\centering
\includegraphics[scale=0.9]{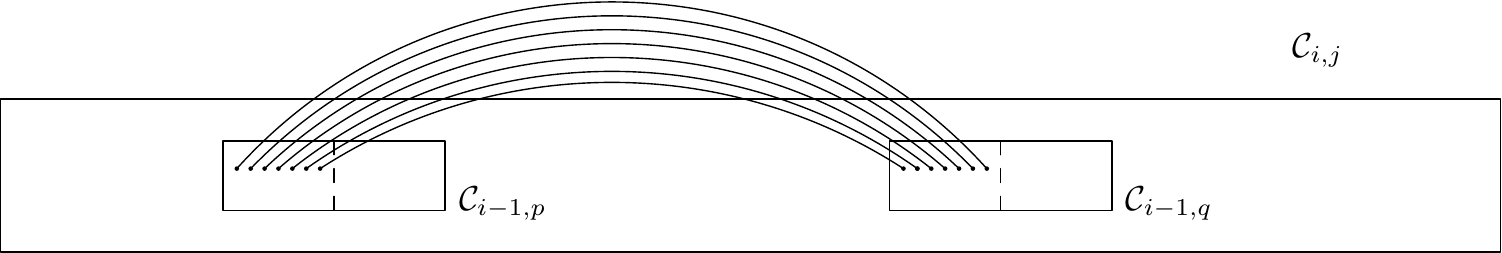}
\caption{An example of a matching between the half-clusters $\mathcal{H}_{i-1,p}^L$ and $\mathcal{H}_{i-1,q}^L$. We add an arbitrary complete matching between them because they are both contained in the same cluster $\mathcal{C}_{i,j}$.}
\label{it-edges}
\end{figure}
\begin{theorem}
For any $\varepsilon>0$ and any $1$-dimensional point set $V$ (with $|V|=n$), the graph $G_\ell$ constructed above is an $\mathcal{O}(k)$-robust $1$-spanner with $\mathcal{O}(n^{1+\varepsilon})$ edges for $\ell\geq\frac{1-\varepsilon}{\varepsilon}$.
\end{theorem}

\begin{proof}
We follow the structure of the proof of Theorem~\ref{nonit-proof}, that is, first we show that the graph constructed above is an $\mathcal{O}(k)$-robust $1$-spanner if $n=(2m)^{\ell+1}$ for some $m\in \mathbb{N}$, then we show the generalization to arbitrary $n$. 
Let us start with the construction of $F^*$. 
To start, set $F_0=F$. Then apply the following rule repeatedly for each layer from bottom to top. Assume that $i$ is the current layer. First, set $F_i = F_{i-1}$. Then, for each half-cluster $\mathcal{H}_{i,j}^L$ in layer $i$, if $|\mathcal{H}_{i,j}^L \cap F_{i-1}| \geq \frac{1}{2}|\mathcal{H}_{i,j}^L |$, we update $F_i$ by adding the clusters $\mathcal{C}_{i,j}$ and $\mathcal{C}_{i,j-1}$. Finally, let $F^*=F_\ell$.
Clearly, $|F_i|\leq 6\cdot|F_{i-1}|$.
Therefore, the size of $F^*$ is at most $6^\ell\cdot |F|$.

Fix $x,y\in V\backslash F^*$ and assume $x<y$. Let $\mathcal{C}_{i+1,j}$ be the smallest cluster that contains both $x$ and $y$. We use induction on the size of the smallest cluster that contains $x$ and $y$ to prove that a path of length $d(x,y)$ between $x$ and $y$ in the subgraph induced by $V\backslash F$ exists. Let $x\in \mathcal{H}_{i,p}^L$ and $y\in \mathcal{H}_{i,q}^R$. There are two cases that we distinguish. The first case is when $q>p+1$ holds, that is, the clusters $\mathcal{C}_{i,p}$ and $\mathcal{C}_{i,q}$ do not intersect. Then, there is a vertex $x'\in \mathcal{H}_{i,p}^R\backslash F^*$ and $y'\in \mathcal{H}_{i,q}^L \backslash F^*$ such that $\{x',y'\}\in E$. By induction there is a path of length $d(x,x')$ between $x$ and $x'$ and a path of length $d(y',y)$ between $y'$ and $y$. Then using the edge $\{x',y'\}$ we get a path of length $d(x,y)$ between $x$ and $y$, see Figure~\ref{1-path-ind}. 
\begin{figure}
\centering
\includegraphics[scale=0.8]{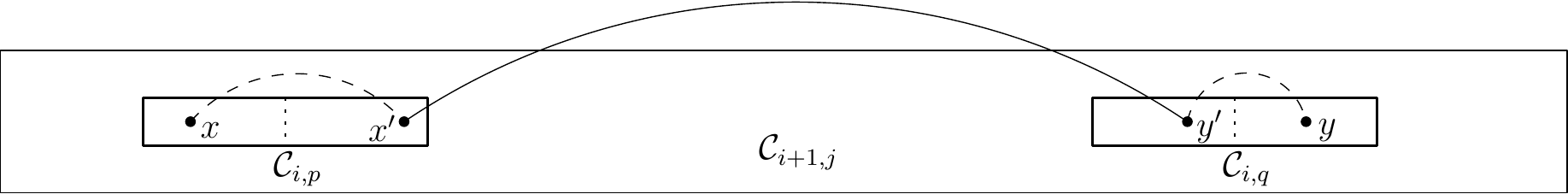}
\caption{The structure of the path between $x$ and $y$ that has length $d(x,y)$ when $q>p+1$ holds.}
\label{1-path-ind}
\end{figure}

The second case is when $q=p+1$ holds, that is, the clusters $\mathcal{C}_{i,p}$ and $\mathcal{C}_{i,q}$ intersect. As it is illustrated in Figure~\ref{1-path-ind2}, there are clusters $\mathcal{C}_{i-1,r}$ and $\mathcal{C}_{i-1,t}$ one layer below that contain the vertices $x$ in the left half and $y$ in the right half, respectively. 
Since $x,y\notin F^*$, therefore $\mathcal{H}_{i,p+1}^L \nsubseteq F^*$ and there exists a cluster $\mathcal{C}_{i-1,s} \subset \mathcal{H}_{i,p+1}^L$ such that $|\mathcal{H}_{i-1,s}^L \cap F^*| < \frac{1}{2}|\mathcal{H}_{i-1,s}^L|$ and $|\mathcal{H}_{i-1,s}^R \cap F^*|<\frac{1}{2}|\mathcal{H}_{i-1,s}^R|$ holds. Therefore, by the pigeonhole principle, we can choose the points $x'\in \mathcal{H}_{i-1,r}^R \backslash F^*$ and $v\in \mathcal{H}_{i-1,s}^L \backslash F^*$ such that $\{ x',v \}\in E$. Similarly, we can choose $z\in \mathcal{H}_{i-1,s}^R \backslash F^*$ and $y'\in \mathcal{H}_{i-1,t}^L \backslash F^*$ such that $\{z,y'\}\in E$. Again, by induction there are paths with length of the Euclidean distance between $x,x'$ and $v,z$ and $y,y'$. Note that $x'=v$ or $y'=z$ can occur, but these cases are simpler. The induction terminates, since at each step the level is decreased at least by one and at the bottom layer all points are connected within one cluster.
\begin{figure}
\centering
\includegraphics[scale=1.2]{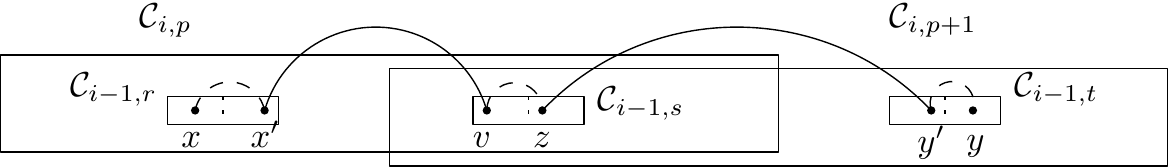}
\caption{The structure of the path between $x$ and $y$ that has length $d(x,y)$ when $q=p+1$ holds.}
\label{1-path-ind2}
\end{figure}

Now to show the generalization to arbitrary $n$. Assume that $(2m)^{\ell+1}<n<(2m+2)^{\ell+1}$. Repeat the same construction on the first $(2m)^{\ell+1}$ points as before. Then extend each layer continuously with half-overlapping clusters as long as they fit. If there are some points left at the end, add one more cluster that has a smaller sized right half. Form matchings between the proper half-clusters as before and form cliques within each cluster in the lowest layer. Regarding the failures, the construction of $F^*$ remains the same. This extension preserves the desired properties and the proof to find a path of length $d(x,y)$ between $x$ and $y$ works the same.

\end{proof}

\noindent
\textit{Acknowledgements.} The work in this paper is supported by the Netherlands Organisation for Scientific Research (NWO) through Gravitation-grant NETWORKS-024.002.003.



\bibliographystyle{plain}
\bibliography{references}
\end{document}